\documentclass[12pt,reqno]{article}

\usepackage[usenames]{color}
\usepackage{amssymb}
\usepackage{graphicx}
\usepackage{amscd}
\usepackage{cite}
\usepackage{enumerate}

\usepackage[colorlinks=true,
linkcolor=webgreen,
filecolor=webbrown,
citecolor=webgreen]{hyperref}

\definecolor{webgreen}{rgb}{0,.5,0}
\definecolor{webbrown}{rgb}{.6,0,0}

\usepackage{color}
\usepackage{fullpage}
\usepackage{float}
\newcommand{\seqnum}[1]{\href{http://oeis.org/#1}{\underline{#1}}}

\usepackage{graphics,amsmath,amssymb}
\usepackage{amsthm}
\usepackage{amsfonts}
\usepackage{epsf}

\DeclareMathOperator{\Fac}{Fac}

\begin{document}

\theoremstyle{plain}
\newtheorem{theorem}{Theorem}
\newtheorem{corollary}[theorem]{Corollary}
\newtheorem{lemma}[theorem]{Lemma}
\newtheorem{proposition}[theorem]{Proposition}
\newtheorem{openquestion}[theorem]{Open Question}

\theoremstyle{definition}
\newtheorem{definition}[theorem]{Definition}
\newtheorem{example}[theorem]{Example}
\newtheorem{conjecture}[theorem]{Conjecture}

\theoremstyle{remark}
\newtheorem{remark}[theorem]{Remark}

\def\g{{\mathbf g}}
\def\u{{\mathbf u}}
\def\U{{\mathbf U}}
\def\v{{\mathbf v}}
\def\w{{\mathbf w}}
\def\t{{\mathbf t}}
\def\T{{\mathbf T}}
\def\f{{\mathbf f}}
\def\F{{\mathbf F}}
\def\G{{\mathbf G}}
\def\ext{{\sf ext}}
\def\cwd{{\sf cwd}}
\def\cwk{{\sf cwk}}
\def\rext{{\sf rext}}
\def\Spec{{\sf Spec}}

\title{Subword complexity and power avoidance}

\author{Jeffrey Shallit\\
School of Computer Science\\
University of Waterloo \\
Waterloo, ON  N2L 3G1 \\
Canada \\
{\tt shallit@uwaterloo.ca} \\
\and 
Arseny Shur\\
Dept. of Algebra and Fundamental Informatics\\
Ural Federal University \\
620000 Ekaterinburg \\
Russia \\
{\tt arseny.shur@urfu.ru} 
}

\maketitle

\begin{abstract}

We begin a systematic study of the relations between subword complexity of infinite words and their power avoidance. Among other things, we show that\\
-- the Thue-Morse word has the minimum possible subword complexity over all overlap-free binary words and all $(\frac 73)$-power-free binary words, but not over all $(\frac 73)^+$-power-free binary words;\\
-- the twisted Thue-Morse word has the maximum possible subword complexity over all overlap-free binary words, but no word has the maximum subword complexity over all $(\frac 73)$-power-free binary words;\\
-- if some word attains the minimum possible subword complexity over all square-free ternary words, then one such word is the ternary Thue word;\\
-- the recently constructed 1-2-bonacci word has the minimum possible subword complexity over all \textit{symmetric} square-free ternary words.

\end{abstract}

\section{Introduction}

Two major themes in combinatorics on words are {\it power avoidance} and {\it subword complexity} (also called {\it factor complexity\/} or just {\it complexity}).  In power avoidance, the main goals are to construct infinite words avoiding various kinds of repetitions (see, e.g., \cite{Rampersad&Shallit:2016}), and to count or estimate the number of length-$n$ finite words avoiding these repetitions (see, e.g., \cite{Shur:2012}).  In subword complexity, the main goal is to find explicit formulas for, or estimate, the number of distinct blocks of length $n$ appearing in a given infinite word (see, e.g., \cite{Ferenczi:1999}). In this paper we combine these two themes, beginning a systematic study of infinite binary and ternary words satisfying power avoidance restrictions. We follow two interlaced lines of research. First, given a power avoidance restriction, we study the set of infinite words satisfying this restriction, focusing on upper and lower bounds on their subword complexity, and examples of words of ``large'' and ``small'' complexity. Second, given a subword complexity restriction, we seek lower bounds on the powers avoided by infinite words of restricted complexity, and for words attaining these bounds. We also tried to cover the remaining white spots with open questions and conjectures. Most of the results are gathered in Table~\ref{tab1}; precise definitions are given below in Section~\ref{ss:def}. 

\begin{table}[!htb] \label{tab1}
\caption{Infinite power-avoiding binary and ternary words of small and large complexity. Question marks indicate conjectured results.}
\begin{center}
\def\arraystretch{1.2}
\begin{tabular}{|l|l|l|l|l|}
\hline
\multicolumn{1}{|c|}{Avoidance restriction}&\multicolumn{1}{|c|}{Small complexity}&\multicolumn{1}{|c|}{Large complexity}\\
\hline
\multicolumn{3}{|c|}{Binary words}\\
\hline
\rule{0pt}{4ex}
all overlap-free&\begin{minipage}{4.7cm} minimum:  Thue-Morse\\ word\end{minipage}&\begin{minipage}{4.5cm} maximum: twisted\\ Thue-Morse word \end{minipage}\\
\hline
\rule{0pt}{4ex}
symmetric overlap-free&\begin{minipage}{4.7cm} minimum:  Thue-Morse\\ word\end{minipage}&
\begin{minipage}{4.5cm} maximum: Thue-Morse\\ word \end{minipage} \\
\hline
\rule{0pt}{4ex}
$(\frac 73)$-power-free&\begin{minipage}{4.7cm} minimum:  Thue-Morse\\ word\end{minipage}& 
\begin{minipage}{4.5cm} maximum: no;\\ upper bound: $<4n$ \end{minipage}\\
\hline
\rule{0pt}{4ex}
symmetric $(\frac 73)$-power-free&\begin{minipage}{4.7cm} minimum:  Thue-Morse\\ word\end{minipage}&
\begin{minipage}{4.5cm} maximum: Thue-Morse\\ word \end{minipage} \\
\hline
\rule{0pt}{4ex}
$(\frac 73)^+$-power-free& \begin{minipage}{4.7cm} Thue-Morse word is not\\ of minimum complexity\end{minipage}& exponential\\
\hline
\rule{0pt}{4ex}
$(\frac 52)^+$-power-free& \begin{minipage}{4.7cm}minimum(?): $2n$\\ (new word)\end{minipage} & \\
\hline
\rule{0pt}{4ex}
$(\frac {5+\sqrt{5}}2)$-power-free& \begin{minipage}{4.7cm}minimum: $n+1$\\ (Fibonacci word)\end{minipage}&  \\
\hline
\multicolumn{3}{|c|}{Ternary words}\\
\hline
\rule{0pt}{4ex}
symmetric $(\frac 74)^+$-power-free& \begin{minipage}{4.7cm}minimal(?) growth:\\ $12n{+}O(1)$ (Arshon word)\end{minipage}&\\
\hline
\rule{0pt}{4ex}
symmetric $(\frac {5+\sqrt{5}}4)$-power-free& \begin{minipage}{4.7cm}minimal\,growth: $6n{+}O(1)$\\ (1-3-bonacci word)\end{minipage}&\\
\hline
\rule{0pt}{4ex}
square-free& \begin{minipage}{4.7cm}minimum(?): ternary\\ Thue word\end{minipage}&exponential\\
\hline
\rule{0pt}{4ex}
symmetric square-free& \begin{minipage}{4.7cm}minimum: $6n-6$\\ (1-2-bonacci word)\end{minipage}&\\
\hline
\rule{0pt}{4ex}
$(\frac 52)$-power-free& \begin{minipage}{4.7cm}minimum(?): $2n+1$\\ (new word)\end{minipage}&\\
\hline
\end{tabular}
\end{center}
\end{table}

The paper is organized as follows. After definitions, we study $\alpha$-power-free infinite binary words for $\alpha\le 7/3$ in Section~2. The number of distinct blocks of length $n$ in this case is quite small and all infinite words are densely related to the Thue-Morse word. We provide answers for most (but not all) questions about complexity of these words. After that, in Section~3 we study low-complexity $\alpha$-power-free infinite binary and ternary words for the case where $\alpha$ is big enough to provide sets of length-$n$ blocks of size exponential in $n$. Here we leave more white spots, but still obtain quite significant results, especially about square-free ternary words. Finally, in Section~4 we briefly consider high-complexity infinite words and relate the existence of words of ``very high'' complexity to an old problem by Restivo and Salemi.

\subsection{Definitions and notation} \label{ss:def}
Throughout we let $\Sigma_k$ denote the $k$-letter alphabet $\{ 0,1,\ldots, k{-}1 \}$. By $\Sigma_k^*$ we mean the set of all finite words over $\Sigma_k$, including the empty word $\varepsilon$. By $\Sigma_k^\omega$ we mean the set of all one-sided right-infinite words over $\Sigma_k$; throughout the paper they are referred to as ``infinite words''.  The length of a finite word $w$ is denoted by $|w|$. 

If $x = uvw$, for possibly empty words $u,v,w,x$, then we say that $u$ is a {\it prefix\/} of $x$, $w$ is a {\it suffix\/} of $x$, and $v$ is a {\it subword\/} (or {\it factor}) of $x$. A prefix (resp., suffix, factor) $v$ of $x$ is {\it proper} if $v \neq x$. A factor $v$ of $x$ can correspond to different factorizations of $x$: $x=u_1vw_1=u_2vw_2=\cdots$. Each factorization corresponds to an \textit{occurrence} of $v$ in $x$; the \textit{position} of an occurrence is the length of the prefix of $x$ preceding $v$. Thus occurrences of $v$ (and, in particular, occurrences of letters) are linearly ordered by their positions, and we may speak about the ``first'' or ``next'' occurrence. An infinite word is \textit{recurrent} if every factor has infinitely many occurrences.

Any map $f: \Sigma_k\to \Sigma_m^*$ ($k,m\ge 1$) can be uniquely extended to all finite and infinite words over $\Sigma_k$ by putting $f(a_0a_1\cdots)=f(a_0)f(a_1)\cdots$, where $a_i\in\Sigma_k$ for all $i$. Such extended maps are called \textit{morphisms}. A morphism is a \textit{coding} if it maps letters to letters.

The Thue-Morse word 
$${\t} = t_0 t_1 t_2 \cdots = 0110100110010110\cdots$$ 
is a well-studied infinite word with many equivalent definitions (see, e.g., \cite{Allouche&Shallit:1999}). The first is that $t_i$ counts the number of $1$'s, modulo $2$, in the binary expansion of $i$.  The second is that $\t$ is the fixed point, starting with $0$, of the morphism $\mu$ sending $0$ to $01$ and $1$ to $10$. 

A {\it language} is a set of finite words over $\Sigma_k$. A language $L$ is said to be {\it factorial\/} if $x \in L$ implies that every factor of $x$ is also in $L$.   If $\u$ is a one-sided or two-sided infinite word, then by $\Fac({\u})$ we mean the factorial language of all finite factors of $\u$. We call a language $L\subseteq \Sigma_k^*$ \textit{symmetric} if $f(L)=L$ for any bijective coding $f: \Sigma_k^*\to \Sigma_k^*$. An infinite word $\u$ is \textit{symmetric} if $\Fac({\u})$ is symmetric. For example, the Thue-Morse word is symmetric.

The so-called {\it subword complexity\/} or {\it factor complexity\/} of an infinite word $\bf x$ is the function $p_{\bf x} (n)$ that maps $n$ to the number of distinct subwords (factors) of length $n$ in $\bf x$.  If the context is clear, we just write $p(n)$. A more general notion is the \textit{growth function} (or \textit{combinatorial complexity}, or \textit{census function}) of a language $L$; it is the function $p_L(n)$ counting the number of words in $L$ of length $n$. Thus, $p_{\bf x}(n)=p_{\Fac({\bf x})}(n)$. These complexity functions can be roughly classified by their \textit{growth rate} $\limsup_{n\to\infty} (p(n))^{1/n}$; for factorial languages, the $\limsup$ can be replaced by $\lim$.  Exponential (resp., subexponential) words and languages have growth rate $>1$ (resp., 1); growth rate 0 implies a finite language. Infinite words constructed by some regular procedure (e.g., those generated by morphisms) usually have small, often linear, complexity (see, e.g., \cite{AlSh03}). 

We say that an infinite word has \textit{minimum} (resp., \textit{maximum}) subword complexity in a set of words $S$ if $p_{\u}(n)\le p_{\v}(n)$ (resp., $p_{\u}(n)\ge p_{\v}(n)$) for every word $\v\in S$ and every $n\ge 0$.

An integer power of a word $x$ is a word of the form $x^n = \overbrace{xx\cdots x}^n$.  If $x$ is nonempty then by $x^\omega$ we mean the infinite word $xxx\cdots$. Integer powers can be generalized to fractional powers as follows:  by $x^\alpha$, for a real number $\alpha \geq 1$, we mean the prefix of length $\lceil \alpha |x| \rceil$ of the infinite word $x^\omega$. If $u$ is a finite word and $x$ is the shortest word such that $u$ is a prefix of $x^\omega$, then the ratio $|u|/|x|$ is called the \textit{exponent} of $u$. The \textit{critical} (or \textit{local}) \textit{exponent} of a finite or infinite word $\u$ is the supremum of the exponents of its factors. Thus, for example, the French word {\tt contentent} (as in {\it ils se contentent}) has a suffix that is the $(3+\sqrt{5})/2 = 2.61803 \cdots$'th power of the word {\tt nte}, as well as the $(8/3)$'th power of this word; the exponent of the word {\tt contentent} is 1, and its critical exponent is $8/3$.

We say a finite or infinite word \textit{is $\alpha$-power-free} or \textit{avoids $\alpha$-powers} if it has no factors that are $\beta$-powers for $\beta \geq \alpha$. Similarly, a finite or infinite word \textit{is $\alpha^+$-power-free} or \textit{avoids $\alpha^+$-powers}  if it has no factors that are $\beta$-powers for $\beta > \alpha$. In what follows, we use only the term ``$\alpha$-power", assuming that $\alpha$ is either a number or  a ``number with a +''. We write $L_{k,\alpha}$ for the language of all finite $k$-ary $\alpha$-power-free words. The growth functions of these power-free languages were studied in a number of papers; see the survey \cite{Shur:2012} for details. For  our study, we need the following rough classification of growth functions of infinite power-free languages \cite{Restivo&Salemi:1985a, Karhumaki&Shallit:2004, KoRa11, ShGo10, TuSh12}: the languages $L_{2,\alpha}$ for $2^+\le\alpha\le 7/3$ have polynomial growth functions, while all other power-free languages are conjectured to have exponential growth functions.  This conjecture has been proved for all power-free languages over the alphabets of even size up to 10 and of odd size up to 101. The ``polynomial plateau'' of binary power-free languages possesses several distinctive properties due to their intimate connection to the Thue-Morse word (see, e.g.,  \cite[Section 2.2]{Shur:2012}).

An infinite word is called \textit{periodic} if it has a suffix $x^\omega$ for some nonempty word $x$; otherwise, it is called \textit{aperiodic}. Obviously, all power-free words are aperiodic. The minimum subword complexity of an aperiodic word is $n+1$, reached by the class of \textit{Sturmian words} \cite{MoHe40}.

A finite word $x$ from a language $L$ is \emph{right extendable} in $L$ if for every integer $n$ there is a word $v$ such that $|v|>n$ and $xv\in L$. Left extendability is defined in a dual way. Further, $x$ is \emph{two-sided extendable} in $L$ if for every integer $n$ there are words $u,v$ such that $|u|,|v|>n$ and $uxv\in L$. We write
$$
\begin{array} {ll}
\rext(L)&=\{x\in L\mid x \text{ is right extendable in } L\}\\
\ext(L)&=\{x\in L\mid x \text{ is two-sided extendable in } L\}
\end{array}
$$
Note that all factors of an infinite word $\u$ are right extendable in $\Fac(\u)$. It is known \cite{Shur08b} that for every language $L$ the languages $\ext(L)$ and $\rext(L)$ have the same growth rate as $L$.

For a word $u$ over $\Sigma_2 = \{ 0, 1 \}$, we say that we \emph{flip} a letter $a$
we replace it with $1-a$. The word $\overline{u}$ obtained from
$u$ by flipping all letters is the \emph{complement} of $u$.





\section{Minimum and maximum subword complexity in small languages}

The $2^+$-power-free words are commonly called \textit{overlap-free} due to the following equivalent characterization: \textit{a word $w$ contains an $\alpha$-power with $\alpha>2$ if and only if two different occurrences of some factor in $w$ overlap}. It is known since Thue \cite{Thue12} that the Thue-Morse word $\t$ is overlap-free. The morphism $\mu$ satisfies the following very strong property.

\begin{lemma}[\cite{Shur00}] \label{l:tmexp}
For every real $\alpha>2$, an arbitrary word $u$ avoids $\alpha$-powers iff the word $\mu(u)$ does.
\end{lemma}

Moreover, all binary $(7/3)$-power-free (in particular, overlap-free) words can be expressed in terms of the morphism $\mu$. Below is the ``infinite'' version of a well-known result proved by Restivo and Salemi \cite{Restivo&Salemi:1985a} for overlap-free words and extended by Karhum{\"a}ki and Shallit \cite{Karhumaki&Shallit:2004} to all  $(\frac 73)$-power-free words. The second statement of this lemma was proved in \cite{RSS11}.

\begin{lemma} \label{l:73fact}
Let $\u$ be an infinite $(7/3)$-power-free binary word, $k\ge 0$ be an integer. Then $\u$ is uniquely representable in the form
\begin{equation} \label{e:73fact}
\u = x_o \mu (x_1\mu(\cdots x_k\mu(\v)\cdots))= x_0\mu(x_1)\cdots \mu^k(x_k)\mu^{k+1}(\v),
\end{equation}
where $\v$ is also an infinite $(7/3)$-power-free binary word, $x_0,\ldots,x_k \in \{ \varepsilon, 0, 1, 00, 11 \} $. Moreover, for every $i\ge 1$ the condition $|x_i|=2$ implies either $|x_{i-1}|=0$, or $|x_{i-1}|=\cdots =|x_0|=1$, or $|x_{i-1}|=\cdots =|x_j|=1$, $|x_{j-1}|=0$ for some $j\in\{1,\ldots,i{-}1\}$.
\end{lemma}

The factorization \eqref{e:73fact} implies that an infinite $(7/3)$-power-free binary word contains the words $\mu^k(0)$ and $\mu^k(1)$ as factors, for all $k\ge 0$. So we immediately get two corollaries of Lemma~\ref{l:73fact}.

\begin{corollary} \label{c:73TM}
Every $(7/3)$-power-free infinite binary word contains, as factors, all elements of $\Fac(\t)$.  In particular, this is true of every overlap-free word.
\end{corollary}

\begin{corollary} \label{c:TM}
The Thue-Morse word $\t$ has the minimum subword complexity among all
binary $(7/3)$-power-free (in particular, overlap-free) infinite words.
\end{corollary}

The subword complexity $p_{\bf t}$ of the Thue-Morse sequence $\t$ has been known
since the independent work of Brlek \cite{Brlek:1989}, de Luca and Varricchio \cite{deLuca&Varricchio:1989a}, and Avgustinovich \cite{Avgustinovich:1994}.  For $n \geq 2$ it is as follows:
\begin{equation} \label{e:complTM}
p_{\t} (n+1) = \begin{cases}
	4n - 2^i  , & \text{if $2^i \leq n \leq 3 \cdot 2^{i-1}$}; \\
	2n + 2^{i+1} , & \text{if $3 \cdot 2^{i-1} \leq n \leq 2^{i+1}$}.
	\end{cases}
\end{equation}
We now consider the analogue of the Thue-Morse sequence, where for $n \geq 0$ we count the number of $0$'s, mod 2, (instead of the number of $1$'s, mod 2) in the binary representation of $n$. By convention, we assume that the binary expansion of 0 is $\varepsilon$. We call this word 
\begin{multline} \label{e:twisted}
\t' = 001001101001011001101001100101101001011\cdots = \\ 00\mu(1)\mu^2(0)\cdots\mu^{2n}(0)\mu^{2n+1}(1)\cdots
\end{multline}
the \emph{twisted Thue-Morse word}. The word $\t'$ was mentioned in \cite{Shallit:2011} and has appeared previously in the study of overlap-free and $(7/3)$-power-free words \cite{DSS15}.  It is the image, under the coding
$\{0, 2\} \rightarrow 0$, $1 \rightarrow 1$, of the
fixed point of the morphism
$0 \rightarrow 02$,
$1 \rightarrow 21$,
$2 \rightarrow 12$, and
is known to be overlap-free.

We now state one of our main results.  The proof follows after a series of preliminary statements.

\begin{theorem} \label{t:twisted}
The twisted Thue-Morse word $\t'$ has maximum subword complexity among all overlap-free infinite binary words, and is the unique word with this property, up to complement. 
\end{theorem}

\begin{remark}
The word $\t'$ has linear subword complexity, as is proved below, and so contains, as factors, only a small fraction of all right-extendable overlap-free words: the number of such words has superlinear growth (see \cite{Kob88}). This fact is explained in a broader context in Theorem~\ref{t:73max}.
\end{remark}

\begin{remark}
The uniqueness of $\t'$, stated in Theorem~\ref{t:twisted}, differs strikingly from the situation with minimum complexity, described by Corollary~\ref{c:TM}. Namely, the language of factors $\Fac(\t)$, and thus the subword complexity $p_{\t}(n)$, is shared by a continuum of infinite words. The explicit construction of all such words can be found in \cite{Shur00}. (Precisely, Section 2 of \cite{Shur00} describes two-sided infinite words with the language of factors $\Fac(\t)$, but all their suffixes have exactly the same factors.)
\end{remark}

A word $w$ is \emph{minimal forbidden} for a factorial language $L$ if $w\notin L$, while all proper factors of $w$ belong to $L$. Every overlap-free infinite word $\u$ having a factor not in $\Fac(\t)$ contains a minimal forbidden word $x$ for $\Fac(\t)$; moreover, every such word $x$ appearing in $\u$ is right extendable in the language $L_{2,2^+}$. These forbidden words are classified in the following lemma.

\begin{lemma}[\cite{Shur05}] \label{l:forbTM}
Let $a_k$ be the last letter of $\mu^k(0)$. The minimal forbidden words for $\Fac(\t)$ are exactly the following words and their complements:
\begin{itemize}
\item[(a)] $000$;
\item[(b)] $r_k = a_k \mu^k(010) 0$, $k \geq 0$;
\item[(c)] $s_k = a_k \mu^k(101) 0$, $k \geq 0$.
\end{itemize}
Among these, only the words $r_k$ (and their complements) are right extendable for $L_{2,2^+}$.  
\end{lemma}

\begin{remark}
This claim is also easy to prove using the {\tt Walnut} prover \cite{M16}.
\end{remark}

\begin{lemma} \label{l:rk}
Let $\u$ be an infinite overlap-free binary word and let $k\ge 0$. If $\u = x r_k \cdots$ or $\u =x \bar r_k \cdots$ for some word $x$, then $|x| \leq 2^k - 1$. In particular, the words $r_k, \bar r_k$ have, in total, at most one occurrence in $\u$.   
\end{lemma}

\begin{proof}
By induction on $k$. First, consider the case $k = 0$.  We have $r_0 = 00100$. Since $\u$ is overlap free, the letter following $xr_0$ in $\u$ is 1. On the other hand, if $x$ is nonempty, it must end with either $0$ or $1$, and in both cases
$x r_0 1$ has an overlap.  So $x$ must be empty.

Now the induction step.  Assume the claimed result is true for $k' < k$; we prove it for $k$. Let $\u = x r_k \v = x a_k \mu^k (010) 0\v$ and assume $|x| \geq 2^k$. The factorization \eqref{e:73fact} of $\u$ implies $\v=1\mu(\v')$ for some infinite word $\v'$, which is overlap-free by Lemma~\ref{l:tmexp}. Next, note that $x=y\bar a_k$ for some nonempty word $y$. Indeed, $\mu^k (010)$ begins with $0110$. If $x$ ends with $a_k$, then $\u$ contains the factor $a_ka_k0110$, preceded by some letter(s); this certainly creates an overlap. Thus 
$$
\u = y \bar a_k a_k \mu^k (010) 01\mu(\v') = y \mu(\bar a_k \mu^{k-1} (010) 0 \v').
$$
On the other hand, $\u=x_0\mu(\u')$ for some $x_0\in\{\varepsilon,0,00,1,11\}$ and some infinite word $\u'$ by Lemma~\ref{l:73fact}. The uniqueness of factorization in Lemma~\ref{l:73fact} implies $y=x_0\mu(x')$ for some word $x'$ (possibly empty). Observing that $\bar a_k=a_{k-1}$, we finally write 
$$
\u =  x_0 \mu(\u'),\quad \text{where } \u'= x' a_{k-1} \mu^{k-1} (010) 0 \v'.
$$
Applying the inductive hypothesis to $\u'$, we obtain $|x'|\le 2^{k-1}-1$. Having $|x|=|x_0|+2|x'|+1$ and $|x|\ge 2^k$, we obtain $|x'|= 2^{k-1}-1$. Since $|x'|$ has the maximum possible length, the inductive hypothesis guarantees that both words $0\u'$ and $1\u'$ contain overlaps. Then the word $\u'$ has some prefix $uu$ that ends with 0, and also some prefix $vv$ that ends with 1 (e.g., $\u'$ can be a word of the form $001001\cdots$). Hence $\mu(\u')$ has the prefix $\mu(u)\mu(u)$ that ends with 1,  and the prefix $\mu(v)\mu(v)$ that ends with 0. This means that both words $0\mu(\u')$ and $1\mu(\u')$ contain overlaps, and thus $x_0=\varepsilon$. So we have $|x|=2|x'|+1= 2^k-1$. This contradicts our assumption $|x|\ge 2^k$ and thus proves the inductive step.
\end{proof}

Recall that a factor $v$ of $\u$ is (right) $\u$-\emph{special} if both $v0$ and $v1$ are factors of $\u$. The set of all special factors of $\u$ is denoted by $\Spec(\u)$. We will omit $\u$ when it is clear from the context. We use the following familiar fact:

\begin{lemma}
The number $D_{\u}(n)=\Spec(\u)\cap\Sigma_2^n$ is the first difference of the subword complexity of $\u$: $D_{\u}(n)=p_{\u}(n+1)-p_{\u}(n)$.
\end{lemma}

\begin{proof}
Consider the function mapping every word from $\Fac(\u)$ of length $n{+}1$ to its prefix of length $n$. Each special factor of $\u$ of length $n$ has two preimages, while each non-special factor of length $n$ has a single preimage.
\end{proof}

\begin{corollary}
For all $n\ge 1$ and every infinite binary word $\u$ we have
\begin{equation} \label{e:CD}
p_{\u}(n)=2+\sum_{1 \leq i < n} D_{\u}(i).
\end{equation}
\end{corollary}

Since $D_{\u}(0)=1$ for every binary word, below we we restrict our attention to $\{D_{\u}(n)\}$ for $n\ge 1$. For example, 
$$
D_{\t}(n)=
\begin{cases}
4,& \text{if $n=2^k+i$ for some  $k\ge 1$, $i>0$, $i\le 2^{k-1}$};\\ 
2,& \text{ otherwise}\end{cases}
$$
as was first computed in \cite{Brlek:1989}. As an infinite word over $\{2,4\}$, this sequence looks like
\begin{equation} \label{e:DT}
224244224444222244444444222222224\cdots
\end{equation}
where each subsequent block of equal letters is twice the size of the previous block of the same letter.

Let $\u$ be overlap-free. By Corollary~\ref{c:TM}, all $\t$-special factors are $\u$-special, so $D_{\u}(n)\ge D_{\t}(n)$ for all $n$. We call $\u$-special factor
\emph{irregular} if it is not $\t$-special. 

\begin{lemma} \label{l:uspec}
For an overlap-free infinite binary word $\u$, all $\u$-special factors are Thue-Morse factors.
\end{lemma}

\begin{proof}
By Lemmas~\ref{l:forbTM} and \ref{l:rk}, every word from $\Fac(\u)\backslash \Fac(\t)$ occurs in $\u$ only once and thus is not $\u$-special. Hence $\Spec(\u)\subseteq \Fac(\t)$.
\end{proof}

\begin{proof}[Proof of Theorem~\ref{t:twisted}]
Let $\u=0\cdots$ be an overlap-free infinite binary word (the case $\u=1\cdots$ is parallel, so we omit it). We show the following four facts:

\begin{itemize}
\item[(i)] at every position in $\u$, the first occurrence of at most one irregular $\u$-special factor begins;
\item[(ii)] for every $n\ge 0$, there is an irregular $\t'$-special factor $v_n$ with the first occurrence beginning at position $n$ of $\t'$;
\item[(iii)] if $w_n$ is an irregular $\u$-special factor with the first occurrence beginning at position $n$ of an overlap-free word $\u$, then $|w_n|\ge |v_n|$;
\item[(iv)] the inequality $|w_n|\ge |v_n|$ is strict for at least one value of $n$.
\end{itemize}

\noindent (i) Let $v$ be an irregular $\u$-special factor. By Lemma~\ref{l:uspec}, $v\in \Fac(\t)$; but either $v0$ or $v1$ is not a Thue-Morse factor by definition of irregularity. W.l.o.g., $v0\notin \Fac(\t)$. Then some suffix of $v0$ is a minimal forbidden word for $\Fac(\t)$. By Lemma~\ref{l:forbTM}, this suffix equals $r_k$ for some $k\ge 0$. So $v0=v'r_k$, where $|v'|<2^k$ by Lemma~\ref{l:rk}. Thus, the first occurrence of $v$ in $\u$ is followed by 0 (and $v0\notin\Fac(\t)$), while all other occurrences of $v$ are followed by 1 (and $v1\in\Fac(\t)$). Now assume that some proper prefix $w$ of $v$ is also an irregular $\u$-special factor. Since $v\in \Fac(\t)$, the occurrence of $w$ as a prefix of $v$ is not the first occurrence of $w$ in $\u$. Hence the first occurrences of each two irregular $\u$-special factors begin in different positions.

\medskip

\noindent (ii) From \eqref{e:twisted} it is easy to see that for every even $k$ the prefix of $\t'$ of length $5\cdot 2^k$ equals $v\mu^k(0100)$, where $v=0=\mu^k(0)$ for $k=0$ and $v$ has the common suffix $\mu^{k-1}(1)$ with $\mu^k(0)$ for $k>0$. Similarly, for every odd $k$ such a prefix equals $v\mu^k(1011)$, where $v$ has the common suffix $\mu^{k-1}(0)$ with $\mu^k(1)$. According to the above description of the irregular special factors, $\t'$ has first occurrences of irregular special factors beginning at each position:
\begin{equation} \label{e:vn}
\begin{array}{l|l|l}
\text{Factor}&\text{Position}&\text{Length}\\
\hline
0010&  0 & 4\\
\phantom{0}0100110& 1 & 7\\
\phantom{00}10011010010110& 2& 14\\
\phantom{001}0011010010110& 3& 13\\
\quad \vdots&\quad \vdots&\quad \vdots\\
r\mu^k(010), r \text{ is a suffix of } \mu^{k-1}(0)& 2^{k-1}+i,\ 0\le i< 2^{k-1} & 4\cdot 2^k - 2^{k-1} - i\\
\end{array}
\end{equation}

\noindent (iii) If $v$ is an irregular $\u$-special factor and its first occurrence is followed by $a$, then $va$ is not Thue-Morse and thus has, by Lemma~\ref{l:forbTM}, the suffix $r_k$ or $s_k$ for some $k\ge 0$ (because $v$ is Thue-Morse). Hence by Lemma~\ref{l:rk}, $v$ is a suffix of the prefix $r'\mu^k(010)$ of $\u$, where $0<|r'|\le 2^k$. Therefore the factor $v$ of length $3\cdot 2^k+i$ occurs in $\u$ for the first time at position at most $2^k-i$ ($1\le i\le 2^k$). Inverting this, the factor $w_n$ such that $2^k\le n < 2^{k+1}$ ($k\ge 0$) is of length at least $3\cdot 2^{k+1} + 2^{k+1} - n = 2^{k+3} - n$. But this number is exactly the length of $v_n$ for every $n\ge 1$; see \eqref{e:vn}. Since $v_0$ is the shortest irregular special factor, we proved the required statement.

\medskip

\noindent (iv) Note that the lengths of irregular special factors, given in \eqref{e:vn}, and the first letter of $u$ allow one to restore the whole word in a unique way, and this word is $\t'$. Hence, for any other word $\u$ starting with $0$ we
have $|w_n|>|v_n|$ for some $n$.

\medskip
To finish the proof, consider the function $f: \Spec(\u)\to \Spec(\t')$ that maps every $\t$-special factor to itself and every irregular factor $w_n$ to the irregular factor $v_n$. The function $f$ is well-defined by (ii) and injective by (i). By (iii), $\phi$ never increases the length of the word. Finally, (iv) implies that $f$ decreases the length of some $w_{n_0}$. Together, these facts imply that for every $n$,
$$
\sum_{1 \leq i < n} D_{\u}(i) \le \sum_{1 \leq i < n} D_{\t'}(i),
$$
and the inequality is strict for $n=n_0$.  The claim of the theorem is now immediate from \eqref{e:CD}.
\end{proof}

As a consequence of the proof, we can determine a closed form for the subword complexity of $\t'$.

\begin{proposition} \label{p:twisted}
The number of special $\t'$-factors of length $n>0$ is given by the formula
$$
D_{\t'}(n)=
\begin{cases}
4,& 2^{k+1} < n \le 3\cdot 2^k \text{ for some } k\ge 0;\\
3,& n=4 \text{ or } 3\cdot 2^k < n \le 7\cdot 2^{k-1} \text{ for some } k\ge 1;\\
2,& \text{otherwise.}
\end{cases}
$$
\end{proposition}

\begin{proof}
The proposition states that the sequence $\{D_{\t'}(n)\}$ ($n\ge 1$) can be written as the following word over $\{2,3,4\}$:
$$
224344324444332244444444333322224\cdots
$$
Comparing this to \eqref{e:DT}, we see that some 2's have been changed to 3's. This means an additional $\t'$-special factor for each corresponding length, and this factor must be irregular. According to \eqref{e:vn}, the set of all positions of 3's indeed coincides with the set of lengths of irregular $\t'$-special factors, thus proving the proposition. 
\end{proof}

\begin{corollary}
The maximum factor complexity of a binary overlap-free infinite word is the factor complexity of the twisted Thue-Morse word $\t'$ and is given, for $n\ge 4$, by the formula
\begin{equation} \label{e:twcompl}
p_{\t'}(n+1)=
\begin{cases}
4n - 3\cdot 2^{i-2}  , & \text{if }2^i \leq n \leq 3 \cdot 2^{i-1}; \\
3n + 3\cdot 2^{i-2}  , & \text{if }3 \cdot 2^{i-1} \leq n \leq 7 \cdot 2^{i-2}; \\
2n + 5\cdot 2^{i-1} , & \text{if }7 \cdot 2^{i-2} \leq n \leq 2^{i+1}.
\end{cases}
\end{equation}
\end{corollary}

\begin{proof}
Immediate from Theorem~\ref{t:twisted}, Proposition~\ref{p:twisted}, and formula~\eqref{e:CD}.
\end{proof}








A table of the first few values of the subword complexity
of $\t'$ follows:

\begin{center}
\begin{tabular}{c|cccccccccccccccccccc}
$n$ & 0 & 1 & 2 & 3 & 4 & 5 & 6 & 7 & 8 & 9 & 10 & 11 & 12 & 13 & 14 & 15 & 16 & 17 & 18\\
\hline
$p_{\t'} (n) $ &
1&2&4&6&10&13&17&21&24&26&30&34&38&42&45&48&50&52&56\\
\end{tabular}
\end{center}

\subsection{Beyond overlap-free words}

Now we turn to the case of $\alpha$-power-free infinite binary words for arbitrary $\alpha$ from the half-open interval $(2, \frac{7}{3}]$. By Corollary~\ref{c:TM}, the Thue-Morse word and its complement are words of minimum complexity. Theorem~\ref{t:73max} below shows that the asymptotic growth of subword complexity belongs to a very small range. Theorem~\ref{t:73nomax} demonstrates that there is no $(7/3)$-power-free  infinite binary word of maximum complexity. 

\begin{theorem} \label{t:73max}
Let $\alpha\leq 7/3$.  Every infinite binary word $\u$ avoiding $\alpha$-powers has linear subword complexity. Moreover, for every $n>0$ one has $p_{\u}(n)< \frac {6}{5}\cdot p_{\t}(n)$.
\end{theorem}

\begin{proof}
First assume that $\u=\mu^m(\v)$ for some word $\v$ and some $m\ge 0$.  Then $\v$ also avoids $\alpha$-powers by Lemma~\ref{l:tmexp}. According to Lemma~\ref{l:forbTM}, the shortest word that can be a factor of $\v$, but does not occur in $\t$, is either $00100$ or $11011$. Hence every factor of $\u$ that is contained in four consecutive blocks of the form $\mu^m(a)$, $a\in\Sigma_2$, is a factor of $\t$. Thus the shortest factor of $\u$ that is not in $\t$ has the length at least $3\cdot 2^m+2$. (If $\v$ contains $00100$, this factor is $r_k$ from Lemma~\ref{l:forbTM}.) So we have 
\begin{equation} \label{e:mum}
\u=\mu^m(\v) \Longrightarrow p_{\u}(n)=p_{\t}(n) \text{ for every }
n=0,1,\ldots,3{\cdot} 2^m{+}1.
\end{equation}
Now let $\u$ be arbitrary. Still, $\u$ satisfies \eqref{e:mum} with $m=0$, so $p_{\u}(n)=p_{\t}(n)$ for $n\le 4$. So we take an arbitrary $n\ge 5$ and choose a unique integer $m\ge 1$ satisfying the condition $3\cdot 2^{m-1}+1 < n \le 3\cdot 2^m+1$. Consider the factorization of $\u$ of type \eqref{e:73fact}:
$$
\u=x_0\mu(x_1\mu(\cdots x_{m-1}\mu(\v)\cdots ))=
x_0\mu(x_1)\cdots \mu^{m-1}(x_{m-1})\mu^m(\v) .
$$
By \eqref{e:mum}, all factors of $\mu^m(\v)$ of length $n$ are Thue-Morse factors, so we have 
$$
p_{\u}(n)-p_{\t}(n)\le |x_0\mu(x_1)\cdots \mu^{m-1}(x_{m-1})|=
\sum_{i=0}^{m-1} 2^i|x_i|.
$$
This upper bound is a bit loose; to tighten it, consider $x_{m-1}$. By Lemma~\ref{l:73fact}, $|x_{m-1}|\le 2$. Let $|x_{m-1}|= 2$ (w.l.o.g., $x_{m-1}=00$). Then
$$
\mu^{m-1}(x_{m-1})\mu^m(\v) = \mu^{m-1}(0)\mu^{m-1}(0100110)\cdots
$$
Since 0100110 is a Thue-Morse factor, so is $\mu^{m-1}(0100110)$. Hence the number of length $n$ words in $\Fac(\u)\backslash \Fac(\t)$ is at most $2^{m-1}+\sum_{i=0}^{m-2} 2^i|x_i|$. Applying the second statement of Lemma~\ref{l:73fact} to $x_{m-1}$, we easily obtain $\sum_{i=0}^{m-2} 2^i|x_i|\le \sum_{i=0}^{m-2} 2^i< 2^{m-1}$. Therefore,
\begin{equation} \label{e:pupt}
p_{\u}(n)-p_{\t}(n)<  2^{m}. 
\end{equation}
Next let $x_{m-1}=0$ and let $\v=abc\cdots$, $a,b,c\in\{0,1\}$. If $a=b=c$ or $a=b=1$, then $\u$ is not $(7/3)$-power free. Otherwise, $1abc$ is a Thue-Morse factor, as well as the suffix $\mu^{m-1}(0)\mu^m(abc)$ of $\mu^m(1abc)$. Then all length-$n$ words of $\u$, that are not in $\Fac(\t)$, begin in $\u$ on the left of $\mu^{m-1}(x_{m-1})$. But $\sum_{i=0}^{m-2} 2^i|x_i|<2^m$, so we again have \eqref{e:pupt}. By the same reason we obtain \eqref{e:pupt} in the case $x_{m-1}=\varepsilon$.

Finally, we apply \eqref{e:complTM} to get 
\begin{equation} \label{e:pt}
p_{\t}(n)\ge p_{\t}(3\cdot 2^{m-1}+2)= 2\cdot (3\cdot 2^{m-1}+1)+2^{m+1}=5\cdot 2^m+2
\end{equation}
and compare \eqref{e:pupt} to \eqref{e:pt} to obtain the required inequality.
\end{proof}

From the proof of Theorem~\ref{t:73max} we have the following: for arbitrary $(7/3)$-power-free infinite words $\u$,
\begin{enumerate}[(a)]
\item if $n\le 4$, then $p_{\u}(n)=p_{\t}(n)$;
\item if $5\le n\le 7$, we write $\u=x_0\mu(\v)$ and see that $p_{\u}(n)\le p_{\t}(n)+1$; the only length $n$ factor of $\u$ that is possibly not in $\Fac(\t)$ is the prefix of $\u$;
\item if $8\le n\le 13$, we write $\u=x_0\mu(x_1)\mu^2(\v)$ and analyze cases to get $p_{\u}(n)\le p_{\t}(n)+2$ for $n=8,9$ and $p_{\u}(n)\le p_{\t}(n)+3$ for $n=10,11,12,13$. Further, we note that every word $\u_1$ having the property $p_{\u_1}(8)= p_{\t}(8)+2$ and beginning with 0 is of the following form:
$$
\u_1=00\mu(1)\mu^2(01\cdots)=00\,10\,0110\,1001\cdots,
$$
and $p_{\u_1}(n)= p_{\t}(n)+2$ for all $n$ in the considered range (the two ``additional'' factors are those beginning at positions 0 and 1). On the other hand, every word $\u_2$ with the property $p_{\u_1}(10)= p_{\t}(10)+3$, starting with 0, has the form 
$$
\u_2=0\mu(00)\mu^2(101\cdots)=0\,0101\,1001\,0110\,1001\cdots,
$$
and one has $p_{\u_2}(8)= p_{\t}(8)+1$ (the only additional factor 10110010 begins at position 2), $p_{\u_2}(9)= p_{\t}(9)+2$ (the additional factors begin at positions 1 and 2), and $p_{\u_2}(n)= p_{\t}(n)+3$ for $n=10,11,12,13$ (the additional factors begin at positions 0, 1, and 2). We see that $\u_1$ reaches the maximum possible complexity for $n=8$ but not for $n=10$, while $\u_2$ reaches this maximum for $n=10$ but not for $n=8$. 
\end{enumerate}

Thus we have proved the following result:
\begin{theorem} \label{t:73nomax}
There is no $(7/3)$-power-free infinite binary word of maximum complexity.
\end{theorem}

If there is no maximum subword complexity, it makes sense to look at some sort of ``asymptotically maximal'' complexity. For a function $y(n)$ of linear growth, let its \textit{linear growth constant} be $\limsup_{n\to\infty} \frac{y(n)}{n}$. For overlap-free infinite binary words, the maximum subword complexity has linear growth constant $7/2$, as can be easily derived from \eqref{e:twcompl}. From \eqref{e:complTM} we see that such a constant for the Thue-Morse word is $10/3$. By Theorem~\ref{t:73max}, this means that the linear growth constants of all $(7/3)$-power-free infinite binary words are upper bounded by 4.

\begin{openquestion}
What is the maximum linear growth constant for the subword complexity of a $(7/3)$-power-free infinite binary word? Which words have such complexity?
\end{openquestion}

\subsection{The symmetric case}

\begin{theorem}
The only possible subword complexity function of a $(7/3)$-power-free symmetric infinite binary word is the function $p_{\t}(n)$.
\end{theorem}

\begin{proof}
Let $\u$ be a $(7/3)$-power-free infinite binary word such that $p_{\u}(n) \ne p_{\t}(n)$. By Corollary~\ref{c:TM}, $\u$ contains a factor that is not Thue-Morse. By Lemma~\ref{l:forbTM}, $\u$ contains one of the factors $r_k,s_k$, or their complements. Assume that $\u$ contains $r_k=a_k\mu^k(010)0$. Taking the representation \eqref{e:73fact} of $\u$, we see that $\u$ has the suffix $a_k\mu^k(0100\cdots)$. If the word $\mu^k(0100\cdots)$ contains the factor $\bar r_k$, it contains either $\mu^k(111)$ or $\mu^k(0110110)$, which is impossible because $\u$ is $(7/3)$-power free. Thus, $\bar r_k$ cannot occur in $\u$ to the right of an occurrence of $r_k$. Similarly, if $\u$ contains $s_k=a_k\mu^k(101)0$, then the suffix $\mu^k(1010\cdots)$ of $\u$ cannot contain $\bar s_k$ without containing the $(5/2)$-power $\mu^k(10101)$.

Repeating the same argument for the factors $\bar r_k,\bar s_k$ in $\u$ we conclude that $\u$ contains neither $r_k,\bar r_k$ simultaneously, nor $s_k,\bar s_k$ simultaneously. So $\u$ is not symmetric. Hence every $(7/3)$-power-free symmetric infinite binary word has subword complexity $p_{\t}(n)$.
\end{proof}

\section{Small subword complexity in big languages}

In this section we study binary and ternary words. Note the interconnection of the results over $\Sigma_2$ and $\Sigma_3$ through the encodings of words in both directions. To ease the reading, we denote words over $\Sigma_3$ by capital letters and other words by small letters.

Our study follows two related questions about small subword complexity:
\begin{itemize}
\item Given a pair $(k,\alpha)$, how small can the subword complexity of an $\alpha$-power-free infinite $k$-ary word be?
\item Given an integer $k$ and a function $f(n)$, what is the smallest power that can be avoided by an infinite $k$-ary word with a subword complexity bounded above by $f(n)$?
\end{itemize}

\subsection{Ternary square-free words}

Among $\alpha$-power-free infinite ternary words, the most interesting are square-free (= 2-power-free) words, the existence of which was established by Thue \cite{Thue06} and in particular $(\frac 74)^+$-power-free words, because $\alpha=(\frac 74)^+$ is the minimal power that can be avoided by an infinite ternary word, as was shown by Dejean \cite{Dej72}.

Consider the ternary Thue word $\T$ \cite{Thue12}, which is the fixed point of the morphism $\theta$ defined by $0\to 012, 1\to 02, 2\to 1$:
$$
\T=T_1T_2T_3\cdots=012021012102012021020121012021012102012101202102\cdots
$$
This word has critical exponent 2, which is not reached, so $\T$ is square-free. Also, $\T$ has two alternative definitions through the Thue-Morse word. The first definition says that for any $i\ge 1$, $T_i$ is the number of zeroes between the $i$'th and $(i{+}1)$'th occurrences of 1 in $\t$. The second definition is
\begin{equation} \label{e:T}
T_i=\begin{cases}
0, & \text{ if } t_{i-1}t_i=01;\\
1, & \text{ if } t_{i-1}=t_i;\\
2, & \text{ if } t_{i-1}t_i=10.\\
\end{cases}
\end{equation}
The definition \eqref{e:T} easily implies a bijection between the length-$n$ factors of $\t$ and length-$(n{-}1)$ factors of $\T$ for any $n\ge 3$. Hence, $p_{\T}(n)=p_{\t}(n{+}1)$ for all $n\ge 2$, and one can use formula \eqref{e:complTM}. In \cite{FrAv99}, the complexity of $\T$ was computed directly from the morphism $\theta$.

\begin{conjecture}
The ternary Thue word $\T$ has the minimum subword complexity over all square-free ternary infinite words.
\end{conjecture}

The above conjecture is supported by the following result, showing that $\T$ is the only candidate for a square-free word of minimum complexity.

\begin{theorem} \label{t:thue3}
If a word $\U$ has minimum subword complexity over all square-free ternary infinite words, then $\Fac(\U)=\zeta(\Fac(\T))$, where $\zeta$ is a bijective coding. 
\end{theorem}

Before proving this theorem and presenting further results, we need to recall an encoding technique introduced in \cite{Sh10jc} by the second author  as a development of a particular case of Pansiot's encoding \cite{Pansiot:1984c}. In what follows, $a,b,c$ are unspecified pairwise distinct letters from $\Sigma_3$. 
Ternary square-free words contain three-letter factors of the form $aba$, called \emph{jumps} (of one letter over another). Jumps occur quite often: if a square-free word $\u$ has a jump $aba$ at position $i$, then the next jump in $u$ occurs at one of the positions $i{+}2$ ($\u=\cdots abaca\cdots$), $i{+}3$ ($\u=\cdots abacbc \cdots$), or $i{+}4$ ($\u=\cdots abacbab\cdots $). Note that a jump at position $i{+}1$ would mean that $\u$ has the square $abab$ at position $i$,  while no jump up to position $i{+}5$ would lead to the square $bacbac$ at position $i{+}1$. Also note that a jump in a square-free word can be uniquely reconstructed from the previous (or the next) jump and the distance between them.  Thus, 
\begin{itemize}
\item[($\star$)] a square-free ternary word $\u$ can be uniquely reconstructed from the following information: the leftmost jump, its position, the sequence of distances between successive jumps, and, for finite words only, the number of positions after the last jump.
\end{itemize}
The property ($\star$) allows one to encode square-free words by walks in the weighted $K_{33}$ graph shown in Fig.~\ref{k33}. The weight of an edge is the number of positions between the positions of two successive jumps. A square-free word $\u$ is represented by the walk visiting the vertices in the order in which jumps occur when reading $\u$ left to right. If the leftmost jump occurs in $\u$ at position $i>1$, then we add the edge of length $i{-}1$ to the beginning of the walk; in this case the walk begins at an edge, not a vertex. A symmetric procedure applies to the end of $\u$ if $\u$ is finite. By ($\star$), we can omit the vertices (except for the first one), keeping just the weights of edges and marking the ``hanging'' edges in the beginning and/or the end. Due to symmetry, we can omit even the first vertex, retaining all information about $\u$ up to renaming the letters. The result is a word over $\{1,2,3\}$ with two additional bits of information (whether the first/last letters are marked; for infinite words, only one bit is needed). This word is called a \textit{codewalk} of $\u$ and denoted by $\cwk(\u)$. For example, here is some prefix of $\T$ (with first letters of jumps written in boldface) and the corresponding prefix of its codewalk (the marked letter is underlined):
\begin{align*}
\T&=01\pmb{2}02\pmb{1}0\pmb{1}21\pmb{0}201\pmb{2}021\pmb{0}20\pmb{1}2\pmb{1}01\pmb{2}02\pmb{1}0\pmb{1}21\pmb{0}20\pmb{1}2\pmb{1}01\pmb{2}021\pmb{0}20\ \cdots\\
\cwk(\T)&=\underline{2}\quad2\quad1\quad2\quad3\quad3\quad2\quad1\quad2\quad2\quad1\quad2\quad2\quad1\quad2\quad3\quad\quad\cdots
\end{align*}
\begin{figure}[htb]
\centering
\includegraphics[trim=88 708 373 59,clip]{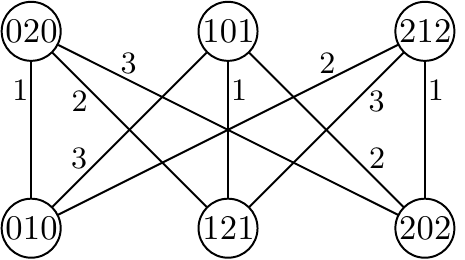}
\caption{\small The graph of jumps in ternary square-free words. Vertices are jumps; two jumps that can follow each other in a square-free word are connected by an edge of weight $i$, where $i$ is the number of positions between the positions of these jumps. Due to symmetry, the graph is undirected.}
\label{k33}
\end{figure}

Note that two words have the same codewalk if and only if they are images of each other under bijective codings and thus have the same structure and the same properties related to power-freeness. A codewalk is \emph{closed} if it corresponds to a closed walk without hanging edges in $K_{3,3}$; e.g., 212212 is closed and 212 is not. 

Clearly, not all walks in the weighted $K_{33}$ graph correspond to square-free words. However, there is a strong connection between square-freeness of a word and forbidden factors in its codewalk. Combining several results of \cite{Sh10jc}, we get the following lemma. (More restrictions can be added to statement 2 of this lemma, but we do not need them in our proofs.)

\begin{lemma}[\cite{Sh10jc}] \label{l:sfwalks}
\leavevmode
\begin{enumerate}
\item If a codewalk has (a) no factors $11, 222, 223, 322, 333$, and (b)~no factors of the form $vabv$, where $v\in\{1,2,3\}^*$, $a,b\in\{1,2,3\}$ and the codewalk $vab$ is closed, then the word with this codewalk is square free.\\
\item The codewalk of a square-free word has no proper factors $11$ and $vav$ for all $v\in\{1,2,3\}^*$, $a,b\in\{1,2,3\}$ such that $va$ is a closed codewalk. In particular, such a codewalk contains no squares of closed codewalks.
\end{enumerate}
\end{lemma}



\begin{proof}[Proof of Theorem~\ref{t:thue3}]
Thue \cite{Thue12} showed that a square-free infinite ternary word contains all six factors of the form $ab$ and all six factors of the form $abc$.  As for the jumps, Thue proved that any two factors \textit{from different parts of the $K_{3,3}$ graph in Fig.}~\ref{k33} can absent. (This is an optimal result since it is easy to see that at least two jumps from each part must be present.) All three possible cases (up to symmetry) with two absent jumps are depicted in Fig.~\ref{k33abc}. We say that a square-free ternary word is of type $i$, $i\in\{1,2,3\}$, if it lacks two jumps connected by an edge of weight $i$. Note that $\T$ has no factors 010 and 212 and thus is of type 3. 

\begin{figure}[htb]
\centering
\includegraphics[trim=88 708 363 59,clip]{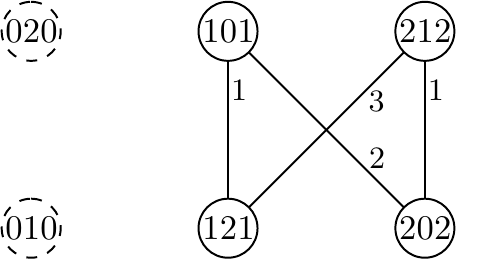}
\includegraphics[trim=78 708 363 59,clip]{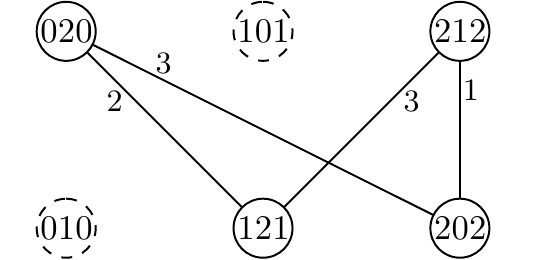}
\includegraphics[trim=78 708 373 59,clip]{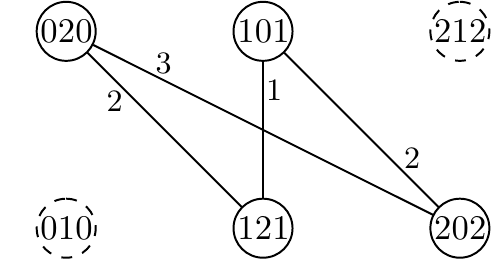}\\
\flushleft
\hskip0.4cm\textbf{a} (type 1)\hskip3.7cm \textbf{b} (type 2)\hskip3.7cm \textbf{c} (type 3)
\caption{\small Avoidance of two jumps in ternary square-free words. ``Type'' is the weight of the edge between the avoided jumps.}
\label{k33abc}
\end{figure}

Assume that a square-free infinite ternary word $\U$ has the minimum subword complexity among all such words. Then $\U$ avoids two jumps, otherwise $p_{\U}(3)> p_{\T}(3)$. W.l.o.g., the two missing jumps are those indicated in Fig.~\ref{k33abc}, depending on the type of $\U$ (if this is not the case, we replace $\U$ with its image under an appropriate bijective coding). Note that all four remaining jumps occur in $\U$ infinitely often, since any suffix of $\U$ is square-free and thus must contain four jumps. Hence $\U$ has eight factors of length 4, containing a jump. Let us compute $p_{\U}(4)$. For this, we need to consider the factors without jumps. First let $\U$ have type 1. 1021 and 2012 are factors of $\U$ because they are the only right extensions of 102 and 201, respectively. If the factor 0120 is absent, then it is impossible to move by the edge of weight 2 from the vertex 101 to the vertex 202; since the codewalk of $\U$ cannot contain 11 or 333 by Lemma~\ref{l:sfwalks} this means that the codewalk of $\U$ always passes the cycle in Fig.~\ref{k33abc}a in the same direction. But this means that the codewalk of $\U$ contains squares of closed walks, which is also impossible by Lemma~\ref{l:sfwalks}. So the word 0120 must be a factor of $\U$. The same argument works for the factors 0210, 1201, 2102, which are responsible, respectively, for the moves from 202 to 101; from 212 to 121; and from 121 to 212. Thus $p_{\U}(4)=14$.

If $\U$ has type 2, the same argument works. Namely, 2012 and 2102 occur in $\U$ as unique extensions of 201 and 210, respectively. The remaining four factors 0120, 0210, 1021, and 2012 are responsible for the moves by the edges of weight 3; if one of them is absent, this leads to squares of closed walks, thus violating of Lemma~\ref{l:sfwalks}. Hence we again have $p_{\U}(4)=14$. The situation changes if $\U$ has type 3: then each of the words 1021, 1201 can occur in $\U$ only as the prefix (1021 is followed by 0, and thus can be preceded neither by 0 nor by 2; similar for 1201). The remaining factors must present by the same argument as in the previous cases. Note that $\T$ avoids 1021 and 1201, and thus $p_{\T}(4)=12$. So $\U$ must have type 3 and $p_{\U}(4)=12$. (Moreover, $\T$ and $\U$ have the same factors up to length 4.) To prove the theorem, it is enough to show that $\Fac(\T)\subseteq\Fac(\U)$; the equality then follows by minimality of complexity of $\U$.

Consider the language $C$ of all finite codewalks that can be read in the graph in Fig.~\ref{k33abc}c and correspond to square-free words. We define two sequences of codewalks by induction:
\begin{align*}
A_0=212&, B_0=3,\\
A_{i+1}=B_i B_i A_i A_i A_i&, B_{i+1}= B_i B_i A_i.
\end{align*}
Below we prove the following three statements, which immediately imply the desired inclusion $\Fac(\T)\subseteq\Fac(\U)$. We include in the preimage $\cwk^{-1}(X)$ only the words of type 3 avoiding the factors 010 and 212, as in Fig.~\ref{k33abc}c.
\begin{itemize}
\item[(i)] $\Fac(\U) \supseteq \cwk^{-1}(A_i)$ for every $i\ge 0$;
\item[(ii)] $\ext(C)=\bigcup_{i\ge 0} \Fac(A_i)$;
\item[(iii)] $\Fac(\T) \subseteq \cwk^{-1}(\ext(C))$.
\end{itemize}
(i) Note that the codewalks $A_iA_i$, $A_iB_i$, and $B_iB_i$ are closed for all $i$. This fact immediately follows by induction from the definition (for the base case, one may consult Fig.~\ref{k33abc}c). 

We prove by induction that for each $i$ some suffix of $\cwk(U)$ is an infinite product of blocks $A_i$ and $B_i$. For the base case note that $\cwk(U)$ has no factors 11 by Lemma~\ref{l:sfwalks}. Hence, as the codewalk reaches one of the vertices 020, 202 (see Fig.~\ref{k33abc}c) it infinitely proceeds between these two vertices with the paths labeled by $A_0$ and $B_0$; each path occurs infinitely many times because  $\cwk(U)$ is aperiodic. The walks $212212, 2123$, and $33$ are obviously closed. 

Now we proceed with the inductive step. First let $i=1$. The factors $A_0B_0A_0$ and $B_0B_0B_0$ does not appear in the codewalks of square-free words by Lemma~\ref{l:sfwalks}. Hence, $B_0$'s always occur in pairs. Further, the factor $B_0A_0A_0B_0$ cannot appear far from the beginning of $\cwk(\U)$, because otherwise it will be uniquely extended to $A_0B_0B_0A_0A_0B_0B_0A_0$, which is the square of a closed codewalk, impossible by Lemma~\ref{l:sfwalks}. Observing that the factor $A_0A_0A_0A_0$ is also forbidden as the square of a closed codewalk, we see that there may be either one or three consecutive blocks $A_0$. Hence some suffix of $\cwk(\U)$ can be partitioned into the blocks $B_0B_0A_0A_0A_0=A_1$ and $B_0B_0A_0=B_1$. As in the base case, both blocks must appear infinitely often to prevent periodicity. For the general case $i>1$ the argument is essentially the same. The only difference is in proving the fact that $A_iB_iA_i$ and $B_iB_iB_i$ are forbidden: now $B_i$ is a prefix of $A_i$ by construction, so these two codewalks are extended to the right by $B_i$, which gives us the square of the closed codewalk $A_iB_i$ (resp., $B_iB_i$). The inductive step is finished.

Thus we know that $\cwk(\U)$ contains $A_i$ (and even $A_iA_i$) for any $i$. Clearly, $A_i$ is not closed, so it corresponds to a walk from the vertex 020 to 202 or vice versa. Hence $\cwk^{-1}(A_i)$ consists of two words, each one corresponds to a walk in one direction. But $A_i$'s in the factor $A_iA_i$ correspond to walks in opposite directions (the codewalk $A_iA_i$ is closed), so both words from $\cwk^{-1}(A_i)$ are factors of $\U$. Statement (i) is proved.

(ii) As shown in the proof of (i), each $A_i$ occurs infinitely often in $\cwk(\U)$. Hence $\ext(C) \supseteq \bigcup_{i\ge 0} \Fac(A_i)$. For the reverse inclusion, first note the following property implied by the proof of (i). There is a function $f(n)$ such that for every $\U$ of type 3 its suffix, equal to the product of the blocks $A_n$ and $B_n$, starts before the position $f(n)$. Let $V\in\ext(C)$ and let $n$ be such that $|V|\le |B_n|$.  By definition of two-sided extendable word, there exists a square-free infinite word $\U$ of type 3 such that $\cwk(\U)$ contains $V$ at position greater than $f(n)$. Then $V$ is a factor of one of the codewalks $A_nA_n$, $A_nB_n$, $B_nA_n$, $B_nB_n$. In each case, $V$ is a factor of $A_{n+2}$, and we have the desired inclusion $\ext(C) \subseteq \bigcup_{i\ge 0} \Fac(A_i)$.

(iii) Since $T$ is recurrent, its codewalk is recurrent as well, implying $\Fac(\cwk(\T))\subseteq \ext(C)$. The result now follows.
\end{proof}

For symmetric words, the square-free ternary infinite word of minimum subword complexity does exist. Recall that the Fibonacci word $\f$ is the fixed point of the binary morphism defined by $\phi(0)=01, \phi(1)=0$. We define the coding $\xi: (0\to 2, 1\to 1)$ and write $\f_{12}=\xi(\f)$. Now consider the \textit{1-2-bonacci word} $\F_{12}\in\Sigma_3^{\omega}$, which is the word beginning with 01 and having the codewalk $\f_{12}$. This word was introduced by Petrova \cite{Pet16}, who proved that $\F_{12}$ has critical exponent $11/6$ (reachable) and no length-5 factors of the form $abcab$. Also, $\F_{12}$ appeared to have a nice extremal property \cite[Proposition 13]{GaSh16}. 

\begin{theorem} \label{t:12bonacci}
The 1-2-bonacci word $\F_{12}$ has the minimum subword complexity over all symmetric square-free ternary infinite words. This complexity equals $6n-6$ for all $n\ge 2$.
\end{theorem}

\begin{proof}
Let $\u$ be a symmetric square-free ternary infinite word. Since $\u$ is aperiodic, it has a special factor of length $n$ for each $n\ge 0$. (In the ternary case, a word $v$ is called a $\u$-special factor if at least two of the words $v0,v1,v2$ are factors of $\u$.) If $v$ is $\u$-special, then the word $\zeta(v)$, where $\zeta$ is any bijective coding, is $\u$-special as well, because of the symmetry of $\u$. Thus, there are at least six $\u$-special factors of length $n$ for each $n\ge 2$. Together with the fact $p_{\u}(2)=6$ mentioned above, this gives the lower bound for the complexity of $\u$: $p_{\u}(n)\ge 6n-6$ for all $n\ge 2$. So we are going to prove that the 1-2-bonacci word is symmetric and its complexity matches this lower bound.

The codewalk $\f_{12}$ of the 1-2-bonacci word $\F_{12}$ has no 3's and thus corresponds to a walk in the subgraph of the $K_{33}$ graph (see Fig.~\ref{k33_12}). By definition of $\F_{12}$, this walk begins at the vertex 010. Let $v$ be a factor of $\f_{12}$. Let us write $f_i=\xi(\phi^i(0))$ for all $i\ge 0$. Then there is $i$ such that $f_i=uvw$ for some words $u$ and $w$. Note that $f_i$ occurs in $\f_{12}$ infinitely often, in particular, as a prefix and after each prefix $f_{i+k}$, where $k>0$. 
\begin{figure}[htb]
\centering
\includegraphics[trim=88 708 373 59,clip]{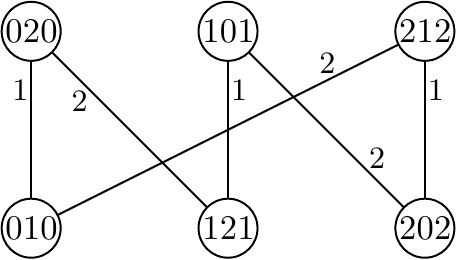}
\caption{\small The graph of jumps in the 1-2-bonacci word.}
\label{k33_12}
\end{figure}

We call two codewalks \textit{equivalent} (and write $u\sim v$) if the corresponding walks, beginning in the same vertex, end in the same vertex. For example, $21221\sim 2$, because the walk 1221 is closed. Similar to \cite{Pet16}, we observe that
$$
f_i\sim \begin{cases}
21,& \text{ if } i\bmod 6 =1;\\
212,& \text{ if } i\bmod 6 =2;\\
2,& \text{ if } i\bmod 6 =0 \text{ or } 3;\\
12,& \text{ if } i\bmod 6 =4;\\
1,& \text{ if } i\bmod 6 =5.\\
\end{cases}
$$
We note that the paths from the vertex 010, labeled by 21, 212, 2, 12, and 1, end in all five remaining vertices. Thus, when reading the codewalk of $F_{12}$, we read $f_i$ starting from each vertex. Reading the word $u$ from all vertices generates a bijection on the set of vertices; thus, we read $v$ from each vertex. Therefore, $\F_{12}$ contains all six factors with the codewalk $v$. Since $v$ is arbitrary, this means that $\F_{12}$ is symmetric.

\medskip
Finally we compute the subword complexity of $\F_{12}$. Consider a  $\F_{12}$-special factor $V$ such that $|V|\ge 5$ and $Va,Vb\in \Fac(\F_{12})$. The codewalk of $V$ has the form $xv\underline{1}$ for some $v\in\Fac(\f_{12})\backslash \{\varepsilon\}$, $x\in\{\varepsilon,\underline{1}, \underline{2}\}$. The factors $Va$ and $Vb$ of $\F_{12}$ have codewalks $xv1$ and $xv\underline{2}$. Hence both $v1$ and $v2$ are factors of $\f_{12}$, so $v$ is $\f_{12}$-special. If $U$ is another $\F_{12}$-special factor of length $|V|$, then, similarly, its codewalk is of the form $yu\underline{1}$ for some $\f_{12}$-special factor $u$ and $y\in\{\varepsilon, \underline{1}, \underline{2}\}$. Since $\f_{12}$ is a Sturmian word, it has only one special factor of each length. But every suffix of a special factor is special, so w.l.o.g.\ $u$ is a suffix of $v$. 

Let $v=v'u$ and assume $v'$ nonempty. The words $U$ and $V$ have suffixes of equal length encoded by $u\underline{1}$. Note that $v'$ encodes at least two letters of $V$ (just two if $v'=1$), while $y$ encodes at most two letters of $U$ (just two if $y=\underline{2}$). But $\f_{12}\in\{1,2\}^\omega$, so if $y=\underline{2}$ then $2u$ is $\f_{12}$-special and thus a suffix of $v$. Hence in this case the last letter in $v'$ is 2, implying that $v'$ encodes at least three letters of $V$. Therefore, $|V|>|U|$ in all cases, contradicting the choice of $U$. Then $v'=\varepsilon$ and $u=v$. Thus we proved that two $\F_{12}$-special factors of the same length have the same codewalk. Due to symmetry, this means that $\F_{12}$ has exactly six special factors of every length $n\ge 5$. Computing  $p_{\F_{12}}(2)=6$, $p_{\F_{12}}(3)=12$, $p_{\F_{12}}(4)=18$, $p_{\F_{12}}(5)=24$, we then have $p_{\F_{12}}(n)=6n-6$ for all $n\ge 2$, as desired.
\end{proof}

\begin{remark} \label{r:13bonacci}
As demonstrated in the proof of Theorem~\ref{t:12bonacci}, 6 is the minimal linear growth constant of an aperiodic symmetric infinite ternary word. However, there exist such words with linear growth constant 6 and critical exponent smaller than $11/6$. An example of such word is the 1-3-bonacci word $\F_{13}$ obtained similar to the 1-2-bonacci word: take the Fibonacci word $\f$, replace all 0's with 3's to get the codewalk $\f_{13}$ and take the word with this codewalk as $\F_{13}$. The critical exponent of $\F_{13}$ is $\frac {5+\sqrt{5}}4\doteq 1.8090\cdots$ \cite{Pet16}; the fact that $\F_{13}$ is symmetric and the equality $p_{F_{13}}(n)=6n$ for all $n\ge 5$ can be proved as in Theorem~\ref{t:12bonacci}.
\end{remark}

Remark~\ref{r:13bonacci} suggests the following question.

\begin{openquestion} \label{o:symm6}
What is the minimal critical exponent of a symmetric infinite ternary word with linear growth constant 6?
\end{openquestion}

As to the case of $(\frac 74)^+$-power-free words, our knowledge is quite limited. The two well-known such words are the Arshon word \cite{Ars37} and the Dejean word \cite{Dej72}. Both are symmetric and recurrent; the codewalk of the Arshon word is the concatenation of some finite prefix and a morphic image of the fixed point of the morphism $\eta$ defined by $\eta(0)=010, \eta(1)=011$ (see \cite[Lemma 2]{PeSh12}). This fixed point has subword complexity $2n$ for every $n\ge 1$. Thus, the Arshon word can be easily proved, similar to the proof of Theorem~\ref{t:12bonacci}, to have linear growth constant 12. We have not computed such a constant for the Dejean word, but there is a numerical evidence that it is not smaller than 12. Thus if the answer to Open Question~\ref{o:symm6} is greater than $7/4$, which looks plausible, then the linear growth constant of the Arshon word is the minimum possible. Two further questions are natural.

\begin{openquestion} 
Is there a symmetric $(\frac 74)^+$-power-free infinite ternary word of minimum subword complexity? If yes, is the Arshon word an example of such word?
\end{openquestion}

\begin{openquestion} 
What is the minimal linear growth constant of a $(\frac 74)^+$-power-free infinite ternary word?
\end{openquestion}

\subsection{Other ternary words} \label{ss:other3}

When the squares in ternary words are allowed, we can build words of smaller complexity. Here we consider the special case of words with the complexity upper bounded by the function $2n+1$. Let
\begin{equation} \label{e:G}
\G= 012020102012010201202\cdots
\end{equation}
be the fixed point of the morphism $\gamma$ defined by $0\to 01, 1\to 2, 2\to 02$. The word $\G$ is a recoding of the sequence \seqnum{A287104} from Sloane's {\it Encyclopedia}.

\begin{lemma} \label{l:G248}
$p_{\G}(n)=2n+1$ for all $n\ge 0$.
\end{lemma}

\begin{proof}
We use standard techniques (see, e.g., \cite{Cas97a}), so we try to keep the proof short. It is sufficient to prove that there are exactly two $\G$-special words of length $n$ for all $n\ge 1$. For small $n$, one can check by exhaustive search that there is a unique special word ending in $0$ and a unique such word ending in $1$ (a $\G$-special word cannot end in 2 because $21,22\notin\Fac(\G)$). Let $V0$ be special; then $V01,V02\in\Fac(\G)$, implying that $\gamma(V)012,\gamma(V)0102\in\Fac(\G)$ and thus $\gamma(V)01$ is special. Similarly, if $U1$ is special, then $U10,U12\in\Fac(\G)$; $\gamma(U)201,\gamma(V)202\in\Fac(\G)$ and thus $\gamma(U)20$ is special. Since each suffix of a special word is special, there are special words of every length ending in 0 and in 1. Now assume that the lemma is false; then for some $n$ one has $D_{\G}(n)>2$, $D_{\G}(1) = \cdots = D_{\G}(n-1)=2$.

Some case analysis is needed; all cases are similar, so we consider one of them. Assume that two special words of length $n$ end with 0. Since their suffixes are special, and only one special word of length $n-1$ ends with 0, these two words are $aV0$ and $bV0$, where $a,b\in\Sigma_3$. Let $a=0,b=1$ (the other case is $a=1,b=2$). Then we can write $V=2V'$. We have 
\begin{align*}
02V'01,\ 02V'02,\ 12V'01,\ 12V'02&\in\Fac(\G) \text{\quad and hence}\\
2\gamma^{-1}(V')0,\ 2\gamma^{-1}(V')2,\ 1\gamma^{-1}(V')0,\ 1\gamma^{-1}(V')2&\in\Fac(\G)
\end{align*}
Then $2\gamma^{-1}(V'),1\gamma^{-1}(V')$ are two special words of length $<n$, ending with the same letter 1; this is impossible by the choice of $n$. Studying all cases in the same way, we reach the same contradiction. Thus the lemma holds.
\end{proof}

\begin{theorem} \label{t:G248}
The critical exponent of the word $\G$ is $2+\frac 1{\lambda^2-1}=2.4808627 \cdots$, where $\lambda=1.7548777\cdots$ is the real zero of the polynomial $x^3 -2x^2 +x - 1$. 
\end{theorem}

\begin{proof}
The critical exponent of $\G$ can be computed by Krieger's method \cite{Kri07}. 
We recall the necessary tools suitable for analyzing $\G$ specifically, rather than in full generality.

For a word $w\in\Sigma_k^*$, we let $|w|_a$ denote the number of occurrences of the letter $a$ in $w$. The \textit{Parikh vector} of $w$ is the vector $\vec{P}(w)=(|w|_0,\ldots,|w|_{k-1})$. By \textit{norm} of a vector we mean the sum of its coordinates; so $\lVert\vec{P}(w)\rVert= |w|$. If $w$ is a prefix of $x^\omega$ for some word $x$, we say that $w$ has \textit{period} $|x|$. In this case, all factors of $w$ of length $|x|$ share the same Parikh vector $\vec{P}(x)$, so we can speak about ``Parikh vector of the period''. If $|x|$ is the minimal period of $w$, we call $x$ the \textit{root} of $w$. The exponent of $w$ then can be written as
$$
\exp(w)=\frac{|w|}{|x|}=\frac{\lVert\vec{P}(w)\rVert}{\lVert\vec{P}(x)\rVert}.
$$
The matrix $A_f$ of a morphism $f: \Sigma_k^*\to \Sigma_m^*$ is a nonnegative integer $k\times m$ matrix, the $i$'th row of which is the Parikh vector of $f(i{-}1)$,  where $i=1,\ldots,k$. For example, the morphism $\gamma$ has the matrix 
$$
A=A_\gamma=\begin{pmatrix}
1&1&0\\0&0&1\\1&0&1
\end{pmatrix}
$$
One has $\vec{P}(f(w))=\vec{P}(w)\cdot A_f$. Note that the characteristic polynomial of $A$ is $x^3-2x^2+x-1$, so the maximal (and unique) real eigenvalue of $A$ is $\lambda$.

A \textit{run} in a finite or infinite word $w$ is an occurrence of a factor $v$ of $w$ such that (a) $\exp(v)\ge 2$ and (b) this occurrence cannot be extended in $w$ to a longer factor with the same minimal period. For example, the word $\G$ has run $2020$ at position 2 with period 2, run $02010201$ at position 3 with period 4, and run $2010201201020120$ at position 4 with period 7; see \eqref{e:G}. For $\G$, as for any word containing squares, the critical exponent equals the supremum of exponents of its runs. Since $\G$ is a fixed point of a morphism, its runs can be grouped into infinite series in the following way: 
\begin{itemize}
\item[-] take a run $V$ at position $i$ with root $X$ (let $\G=UV\cdots$, $|U|=i$); 
\item[-] take the occurrence of $\gamma(V)$ at position $|\gamma(U)|$ and extend it to a run with period $|\gamma(X)|$;
\item[-] take the obtained run as $V$ and repeat.
\end{itemize}
For example, the runs mentioned above form the beginning of a series:
$$
\begin{array}{lclclcl}
2020&\to& 02010201&\to& \pmb{2}01020120102012\pmb{0}&\to&\cdots\\
\text{at }2&&\text{at }3&&\text{at }4\\
U=01&&\gamma(U)=012&&\gamma^2(U)=01202\\
X=20&&\gamma(X)=0201&&\gamma^2(X)=0102012\\
&&\text{no extensions}&& \text{extended left by }2, \text{right by }0
\end{array}
$$
In the case of $\G$, the left (resp., right) extension is a longest common suffix (resp., prefix) of $\gamma$-images of corresponding letters. Thus, the left extension is either $2$ (the common suffix of $\gamma(1)$ and $\gamma(2)$) or $\varepsilon$, and the right extension is either $0$ (the common prefix of $\gamma(0)$ and $\gamma(2)$) or $\varepsilon$. It is possible to compute the exponents of any run in a series in a uniform way. Namely, if one has a series $\{V_m\}_0^\infty$ such that $V_0$ has root $X$, then the exponent of each $V_m$ can be computed by the following formula:
\begin{equation} \label{e:krieger}
\exp(V_m)=
\frac{\lVert\vec{P}(V_0)\cdot A^m+\sum_{i=1}^{m}\vec{P}_i\cdot A^{m-i}\rVert}
{\lVert\vec{P}(X)\cdot A^m\rVert}\ ,
\end{equation}
where $\vec{P}_i$ is the sum of the Parikh vectors of the words which were added to $\gamma(V_{i-1})$ to get $V_i$. Consider the series introduced above, with $V_0=2020$. From definition of $\gamma$ it is easy to find that $\vec{P}_i = (0,0,0)$ for odd $i$ and $\vec{P}_i = (1,0,1)$ for even $i$. Observing that $\vec{P}(X) = (1,0,1)$, $\vec{P}(V_0) = 2\vec{P}(x)$, we simplify \eqref{e:krieger} to get
\begin{equation} \label{e:krieger2}
\exp(V_m)= 2+
\frac{\lVert (1,0,1)\cdot \big(A^{m-2}+ A^{m-4}+\cdots +A^{m-2\lfloor m/2\rfloor} \big)\rVert}
{\lVert (1,0,1)\cdot A^m\rVert}\ .
\end{equation}
One can check that the sequence $\{\exp(V_m)\}_0^\infty$ is strictly increasing, so that its supremum equals its limit. The limit can be computed by standard machinery of Perron-Frobenius theory. We describe the idea, omitting the plain calculus. For large $m$, the vector $(1,0,1)\cdot A^m$ is very close to the eigenvector of the matrix $A$ corresponding to its maximal eigenvalue $\lambda$; so the multiplication of this vector by $A$ corresponds, up to a small error, to its multiplication by $\lambda$. Next note that if we multiply the matrix in the numerator of \eqref{e:krieger2} by $(A^2-I)$, where $I$ is the identity matrix, we obtain either $A^m-I$ or $A^m-A$, depending on the parity of $m$. Hence, as $m\to\infty$, the numerator multiplied by $(\lambda^2-1)$ approaches the denominator. Therefore, 
$$
\lim_{m\to\infty} \exp(V_m)=2+\frac 1{\lambda^2-1},
$$
as required. In the same way, we consider the series starting with the run 201201 at position 8. For this run, $\vec{P}_i = (1,0,0)$ for odd $i$ and $\vec{P}_i = (0,0,1)$ for even $i$. Thus we get the following equality instead of \eqref{e:krieger2}:
\begin{equation} \label{e:krieger3}
\exp(V_m)= 2+
\frac{\lVert (1,0,0)\cdot \big(A^{m-1}+ A^{m-3}+\cdots \big) +  (0,0,1)\cdot \big(A^{m-2}+ A^{m-4}+\cdots \big)\rVert}
{\lVert (1,1,1)\cdot A^m\rVert}\ .
\end{equation}
Observing that $(1,0,0)A+(0,0,1)=(1,1,1)$, we can reduce \eqref{e:krieger3} to \eqref{e:krieger2}. (The particular nonnegative vector does not matter; what matters is that the vector in the numerator and the denominator is the same.) So we get the same limit as above.

For all other series of runs, the initial run is long enough to ensure that the vector $\sum_{i=1}^{m}\vec{P}_i\cdot A^{m-i}$ has much smaller norm than $\vec{P}(X)\cdot A^m$. 
The theorem is proved.
\end{proof}

\begin{conjecture} \label{jef}
Among all ternary words with subword complexity bounded above by $2n+1$, the word $\G$ has the lowest possible critical exponent.   
\end{conjecture}

\subsection{Binary words}

When we switch from $(\frac 73)$-power-free to $(\frac 73)^+$-power-free infinite binary words, the Thue--Morse word loses its status as the word of minimum complexity. Let us give an example of a $(\frac 73)^+$-power-free infinite word with subword complexity incomparable to $p_{\t}(n)$. Consider the morphism $g: \Sigma_3\to \Sigma_2$ defined by the rules
\begin{align*}
0 &\rightarrow 01100100110\,1001\,0110\,1001\,1001, \\
1 &\rightarrow 01100100110\,1001\,0110\,0110\,1001, \\
2 &\rightarrow 01100100110\,1001\,1001\,0110\,1001.
\end{align*}
It maps square-free ternary words to $(\frac 73)^+$-power-free binary words (see \cite[Section 3]{Shur00}). 

\begin{theorem}
Let $\g=g(\T)$, where the morphism $g$ and the ternary Thue word $\T$ are defined above. Then $p_{\g}(n)<p_{\t}(n)$ for infinitely many values of $n$.
\end{theorem}

\begin{proof}
Let $m=2^k$ for some $k\ge 2$ and compare the number of factors of length $n=27m-7$ in $\g$ and $\t$. Note that images of letters under $g$ differ only by the factor of length 8 at position 15. Hence a length-$n$ factor of $\g$ contains exactly $m$ such ``identifying'' factors (one of them, possibly, only partially). Thus every length-$n$ factor of $\g$ is uniquely indentified by its $g$-preimage of length $m$ and its ininial position inside the $g$-image of a letter. So $p_{\g}(n)\le 27\cdot p_{\T}(m)=27\cdot 3m=81\cdot 2^k$. By \eqref{e:complTM} we have $p_{\t}(n)=2(n-1)+2^{k+5}=86\cdot 2^k-16$. Since $5\cdot 2^k>16$, we obtain $p_{\g}(n)<p_{\t}(n)$.
\end{proof}

\begin{openquestion} \label{o:lessTM}
What is the minimun value of $\alpha$ such that some $\alpha$-power-free infinite binary word $\u$ has smaller complexity than the Thue-Morse word? (Recall that ``smaller'' means $p_{\u}(n)\le p_{\t}(n)$ for all $n$ and $p_{\u}(n)< p_{\t}(n)$ for some $n$.)
\end{openquestion}

\smallskip
It is known \cite{CaLu00} that the critical exponent of a Sturmian word is at least $(5+\sqrt{5})/2 \doteq 3.61803\cdots$, and the minimum is reached by the Fibonacci word $\f$ defined above. 

\begin{remark}
The words of subword complexity $n+c$ for any integer constant $c>1$ were characterized by Cassaigne \cite[Proposition 8]{Cas97c}  as having the form $uf(\v)$, where $u$ is a finite word, $\v$ is a Sturmian word and $f$ is a morphism. So it is not hard to prove that these words have the same minimum critical exponent $(5+\sqrt{5})/2$.
\end{remark}

So, for the critical exponents smaller than $(5+\sqrt{5})/2$ we look at the infinite words with linear growth constant 2. As Theorem~\ref{t:52} below shows, such words can have the critical exponent as small as $5/2$. This gives an upper bound for the value of $\alpha$ in Open Question~\ref{o:lessTM}.

\begin{remark}
By backtracking, one can prove that the longest binary words avoiding $5/2$-powers and with subword complexity $\leq 2n$ are of length $38$. They are
\begin{align*}
00110011010011001001101001100100110010&,\\ 
00110011010011001001101001100100110011&,
\end{align*}
and their reversals and complements.
\end{remark}

\begin{theorem} \label{t:52}
Let $\tau: \Sigma_3 \to \Sigma_2$ be the morphism defined by 
$0 \rightarrow 0 ,\ 1 \rightarrow 01 ,\ 2 \rightarrow 011$
and $\G$ be the word defined in Section~\ref{ss:other3}. Then
$$\tau(\G) = 0010110011001001100101100100110010110011001 \cdots $$
has the lowest critical exponent among all binary words with subword complexity $\leq 2n$.  It (i) avoids $(\frac{5}{2})^+$-powers and (ii) has subword complexity exactly $2n$ for all $n>0$.
\end{theorem}

\begin{proof}
In the proof we refer to the properties of the word $\G$ (Lemma~\ref{l:G248} and Theorem~\ref{t:G248}) and their proofs. For (ii), we check by hand that for small $n$ the word $\tau(\G)$ has two special words, one ending with 0 and the other ending with 1. As in the proof of Lemma~\ref{l:G248}, we then assume that statement  (ii) is false, choose $n$ such that $D_{\tau(\G)}(n)>2$, $D_{\tau(\G)}(n{-}1)=\cdots = D_{\tau(\G)}(1)=2$, and do case analysis. There are two $\tau(\G)$-special length-$n$ words ending with the same letter; these words have common suffix of length $n-1$ by the choice of $n$. Consider the case where this suffix equals $0v0$ for some $v\in\Fac(\tau(\G))$. From the condition
$$
00v00,00v01,10v00,10v01\in\Fac(\tau(\G))
$$
we conclude that $0v=\tau(V)$ for some $V\in \Fac(\G)$, and all words $0V0,0V2,1V0,1V2$ are factors of $\G$ (the second conclusion uses the fact that $\G$ has no factor 21). Then both $0V$ and $1V$ are $\G$-special, contradicting the fact, shown in the proof of Lemma~\ref{l:G248}, that $\G$ has at most one special word of a given length ending with a given letter. The cases where the common suffix has the form $0v1,1v0$, or $1v1$ are similar and imply the same contradiction. Hence $D_{\tau(\G)}(n)=2$ for all $n\ge1$, and then $p_{\tau(\G)}(n)=2n$.

\medskip
For (i), we check that $\tau(\G)$ contains the $(\frac 52)$-power $0110011001$, and all other runs with small periods have smaller exponents. For large periods, we should check the exponents of runs which are obtained as images of runs from  two series described in the proof of Theorem~\ref{t:G248}. The images of runs in $\G$ are extended by runs in $\tau(\G)$ by 1 or 2 letters on the right ($\tau$-images of all letters begin with 0, $\tau$-images of 1 and 2 begin with 01) and by 0 or 1 letters two the left ($\tau$-images of 1 and 2 end with 1). For the series of runs in $\G$ starting with $2020$ exactly two letters are added to each $\tau(V_m)$ to make a run in $\tau(G)$. So if we write $v_m$ for the run in $\tau(\G)$ obtained from the $\tau$-image of $V_m$ and $B$ for the matrix of $\tau$, we get from \eqref{e:krieger2}
$$
\exp(v_m)= 2+
\frac{\lVert (1,0,1)\cdot \big(A^{m-2}+ A^{m-4}+\cdots +A^{m-2\lfloor m/2\rfloor} \big)\cdot B\rVert +2}
{\lVert (1,0,1)\cdot A^m\cdot B\rVert}\ .
$$
It is not hard to check that $\exp(v_0)=5/2$ and $\exp(v_m)< 5/2$ for all $m>0$.  Similar computation can be performed for the series of runs in $\G$ starting with $201201$; here $\exp(v_m)< 5/2$ for all $m$. As a result, we conclude that $\tau(\G)$ has critical exponent $5/2$, as required.
\end{proof}


\section{Large subword complexity in big languages}

If a language $L_{k,\alpha}$ has an exponential growth function, then it seems quite natural that there would be infinite $\alpha$-power-free words over $\Sigma_k$ having exponential subword complexity. For example, Currie and Rampersad \cite[Prop.~9]{Currie&Rampersad:2010} gave an example of a squarefree word over $\Sigma_3$ having exponential subword complexity.

Additional examples of such words can be provided using some standard techniques. Below we give the examples for the minimal binary and minimal ternary power-free languages of exponential growth.

\begin{theorem} \label{t:expex}
\leavevmode
\begin{enumerate}[(a)]
\item There is an infinite binary $(\frac 73)^+$-power-free word having exponential subword complexity.
\item There is an infinite ternary $(\frac 74)^+$-power-free word having exponential subword complexity.
\end{enumerate}
\end{theorem}

\begin{proof}
\leavevmode
\noindent (a) First, create an infinite square-free word over $\Sigma_4$ with exponential subword complexity. For this, take an infinite square-free word $\u$ over $\Sigma_3$ and an infinite word $\v$ of exponential complexity over $\{2,3\}$. For each $i\ge 1$, replace the $i$'th occurrence of  symbol $2$ in $\u$ with the $i$'th symbol of $\v$. The resulting word obviously satisfies the desired properties. Now apply the $21$-uniform morphism $h: \Sigma_4^* \rightarrow \Sigma_2^*$ from \cite[Lemma 8]{Karhumaki&Shallit:2004}. The lemma guarantees that the image of a square-free word is $(7/3)^+$-power-free, and every uniform injective morphism preserves the property of having exponential subword complexity.

\medskip

\noindent (b) Start with an infinite $(7/5)^+$-power-free word over $\Sigma_4$, which exists by Pansiot's result \cite{Pansiot:1984c}. As in (a), replace the occurrences of $3$ in this word by an infinite word over $\{ 3, 4\}$ with exponential subword complexity, getting an infinite $(7/5)^+$-power-free word over $\Sigma_5$ with exponential subword complexity. Now apply the morphism of Ochem \cite[Theorem 4.2]{Ochem:2006}.  The result is guaranteed to be $(7/4)^+$-free and to have exponential subword complexity.
\end{proof}

As usual, in such examples the growth rate of subword complexity is barely above 1. Are there words having the subword complexity comparable to the growth function of the whole language? It turns out that this problem is closely related to an old problem by Restivo and Salemi. In \cite{ReSa85a} they posed the following problem: \textit{given two square-free words $u,v\in\Sigma_3^*$, provide an algorithm deciding whether there is a word $w\in\Sigma_3^*$ such that $uwv$ is square-free}. The problem can be generalized to any language $L_{k,\alpha}$, and is still open except for the case of small binary languages (due to the existence of factorizations of type \eqref{e:73fact}, it is easy to connect any right extendable $u$ to any left extendable $v$ by an appropriate Thue-Morse factor). 

Clearly, the interesting part of the problem is formed by the case where $u$ is right extendable and $v$ is left extendable. To the best of our knowledge, there are no known tuples $(k,\alpha, u,v)$ such that $u$ is right extendable in $L_{k,\alpha}$, $v$ is left extendable in $L_{k,\alpha}$, and no word of the form $uwv$ belongs to $L_{k,\alpha}$. For our purposes, we restrict ourselves to the consideration of two-sided extendable words.

We say that a language $L_{k,\alpha}$ \textit{has the Restivo-Salemi property} if for every $u,v\in\ext(L_{k,\alpha})$ there is a word $w$ such that $uwv\in\ext(L_{k,\alpha})$.

\begin{theorem} \label{t:bigmax}
A power-free language $L_{k,\alpha}$ has the Restivo-Salemi property if and only if all words from $\ext(L_{k,\alpha})$ are factors of some $\alpha$-power-free infinite recurrent $k$-ary word $\u$.
\end{theorem}

An already mentioned result of \cite{Shur08b} says that $L$ and $\ext(L)$ have the same growth rate, so Theorem~\ref{t:bigmax} implies the following.

\begin{corollary} \label{c:RS}
If a power-free language $L_{k,\alpha}$ possesses the Restivo-Salemi property, then there is an $\alpha$-power-free infinite $k$-ary word $\u$ with subword complexity having the same growth rate as $L_{k,\alpha}$.
\end{corollary}

\begin{proof}[Proof of Theorem~\ref{t:bigmax}]
For the forward implication we endow $\Sigma_k^*$ with the radix order (the words are ordered by length, and the words of equal length are ordered lexicographically) and build the word $\u$ by induction. As the base case, we build the prefix $u_0=0$. For the inductive step, assume that the prefix $u_n$ was constructed so far. Let $v_n$ be the smallest in radix order word from $\ext(L_{k,\alpha})$ that is not a factor of $u_n$. Then we take $w_n$ such that $u_nw_nv_n\in \ext(L_{k,\alpha})$ and put $u_{n+1}=u_nw_nv_n$. The resulting word $\u$ is $\alpha$-power-free by construction. Further, every word $v\in  \ext(L_{k,\alpha})$ is a factor of some $u_n$ and thus of $\u$. Finally, for an arbitrary $v\in  \ext(L_{k,\alpha})$ and every $n$, there is a word $x$ such that $|x|>|u_n|$ and $xv \in  \ext(L_{k,\alpha})$; since $\u$ contains the factor $xv$, there is an occurrence of $v$ in $\u$ outside the prefix $u_n$. Hence $v$ occurs in $\u$ infinitely many times. Thus $\u$ is recurrent and we proved this implication.

Now turn to the backward implication. For arbitrary words $u,v\in \ext(L_{k,\alpha})$ each of them occurs in $\u$ infinitely often, so we can find a factor of the form $uwv$. This factor also occurs in $\u$ infinitely often, allowing us to find arbitrarily long words $x,y$ such that $xuwvy$ is a factor of $\u$. Then $uwv\in \ext(L_{k,\alpha})$. Hence we proved the Restivo-Salemi property for $L_{k,\alpha}$.
\end{proof}

\begin{remark}
It is worth mentioning that for small binary languages Theorem~\ref{t:bigmax} works in an extremal form. Since $2^+\le\alpha\le 7/3$ implies $\ext(L_{2,\alpha})=\Fac(\t)$, the language $L_{2,\alpha}$ trivially has the Restivo-Salemi property; as we know from Corollary~\ref{c:73TM}, \textit{all} $\alpha$-power-free infinite binary words contain all words of $\ext(L_{2,\alpha})$ as factors.
\end{remark}

The following conjecture is based on extensive numerical studies.

\begin{conjecture} \cite[Conjecture 1]{Sh09dlt} \label{c:RSall}
All power-free languages satisfy the Restivo-Salemi property.
\end{conjecture}

As an approach to Conjecture~\ref{c:RSall}, we suggest the following.

\begin{openquestion}
Prove the converse of Corollary~\ref{c:RS}.
\end{openquestion}

\bibliographystyle{new}
\bibliography{abbrevs,sc2}

\end{document}